\documentclass[journal,12pt,draftclsnofoot,onecolumn]{IEEEtran}
\usepackage{braket}
\usepackage{color}
\usepackage{flafter}   
\usepackage{amsfonts,amssymb,amsmath,calc}
\usepackage{mathrsfs,mathtools}
\usepackage{pstricks}
\usepackage{amsthm}
\usepackage{nicefrac}
\usepackage[english]{babel}
\usepackage[sort&compress,square,comma,numbers]{natbib}
\usepackage{caption}
\usepackage{graphicx}
\usepackage[pdftex]{hyperref}%

\usepackage[affil-it]{authblk}
\usepackage{bm} 
\usepackage{subfig}

\usepackage[T1]{fontenc}

\newcommand{\field}[1]{\mathcal{#1}}
\newcommand{\A}{{\field{A}}}
\newcommand{\B}{{\field{B}}}

\renewcommand{\epsilon}{\varepsilon}
\renewcommand{\theta}{\vartheta}
\renewcommand{\rho}{\varrho}
\renewcommand{\phi}{\varphi}

\renewcommand{\vec}{\bm}

\newcommand{\Prob}[1]{\operatorname{\mathbb{P}} \left(#1\right) }
\newcommand{\Exp}[1]{ {\operatorname{\mathbb{E}} \left[ #1 \right]} }

\newcommand{\ie}{{i.\,e.\ }}
\newcommand{\eg}{{e.\,g.\ }}

\let\originalleft\left
\let\originalright\right
\renewcommand{\left}{\mathopen{}\mathclose\bgroup\originalleft}
\renewcommand{\right}{\aftergroup\egroup\originalright}
\let\Right\right
\let\Left\left
\makeatletter
\def\right#1{\Right#1\@ifnextchar#1{\!}{}}
\def\left#1{\Left#1\@ifnextchar#1{\!}{}}
\makeatother

\DeclarePairedDelimiter\abs{\lvert}{\rvert}%
\DeclarePairedDelimiter\norm{\lVert}{\rVert}%
\makeatletter
\let\oldabs\abs
\def\abs{\@ifstar{\oldabs}{\oldabs*}}
\let\oldnorm\norm
\def\norm{\@ifstar{\oldnorm}{\oldnorm*}}
\makeatother

\theoremstyle{plain}
\newtheorem{theorem}{Theorem}
\newtheorem{problem}{Problem}

\theoremstyle{definition}
\newtheorem{lemma}[theorem]{Lemma}

\newtheorem{definition}{Definition}

\newtheorem{rmk}{Remark}
\newtheorem{model}{Model}

\newtheorem*{rep@theorem}{\rep@title}
\newcommand{\newreptheorem}[2]{%
	\newenvironment{rep#1}[1]{%
		\def\rep@title{#2 \ref*{##1}}%
		\begin{rep@theorem}}%
		{\end{rep@theorem}}}
\newreptheorem{theorem}{Theorem}
\newreptheorem{corollary}{Corollary}

\usepackage{paralist}
\graphicspath{{./figures/}}
\usepackage{xspace}
\usepackage{dsfont}
\usepackage[normalem]{ulem}

\usepackage[affil-it]{authblk}
\newcommand{\knowledge}[1]{\mbox{#1-knowledge}}
\newcommand{\jvec}{{\vec{j}}}

\usepackage[smaller]{acronym}

\definecolor{halfgray}{gray}{0.55} 
\definecolor{webgreen}{rgb}{0,.5,0}
\definecolor{webbrown}{rgb}{.6,0,0}
\definecolor{RoyalBlue}{cmyk}{1, 0.50, 0, 0}
\hypersetup{%
 pdfstartview=FitV,%
unicode=true,bookmarks=true,bookmarksnumbered=true,bookmarksopen=true,bookmarksopenlevel=1,breaklinks=false,%
    pdfauthor={},%
    pdfsubject={},%
    pdfkeywords={},%
    pdfcreator={Alessandro Checco and Carlo Lancia and Doug Leith},%
    pdfproducer={LaTeX},%
    colorlinks=true, linktocpage=true, pdfstartpage=3, pdfstartview=FitV,%
    breaklinks=true, pdfpagemode=UseNone, pageanchor=true, pdfpagemode=UseOutlines,%
    plainpages=false, bookmarksnumbered, bookmarksopen=true, bookmarksopenlevel=1,%
    hypertexnames=true, pdfhighlight=/O,
    urlcolor=webbrown, linkcolor=RoyalBlue, citecolor=webgreen, 
}%

  \usepackage{etoolbox}
  \makeatletter
  \patchcmd{\AC@aclp}{\AC@acl{#1}s}{\AC@acl{#1}{\normalsize s}}{}{}
  \patchcmd{\AC@acsp}{\AC@acs{#1}s}{\AC@acs{#1}{\normalsize s}}{}{}
  \makeatother
\usepackage{siunitx}

\DeclareSIUnit{\dBm}{dBm}
\title{Updating Neighbour Cell List via Crowdsourced User Reports: \\ a Framework for Measuring Time Performance}
\author[1]{A.\ Checco}
\author[2]{C.\ Lancia}
\author[3]{D.J.\ Leith}
\affil[1]{University of Sheffield}
\affil[3]{Trinity College Dublin}
\affil[2]{Leiden University}

\acrodef{ACL}{Adjacent Cell List}
\acrodef{NCLD}{Neighbour Cell List Discovery}
\acrodef{NCL}{Neighbour Cell List}
\acrodef{FK}{Full Knowledge}
\acrodef{AP}{Access Point}
\acrodef{MC}{Markov chain}
\acrodef{DSR}{Detected Set Reporting}
\acrodef{CPICH}{Common Pilot Channel}
\acrodef{UE}{User Equipment}
\newcommand{\bs}{base\-sta\-tion\xspace}
\newcommand{\bss}{base\-sta\-tions\xspace}
\newcommand{\ap}{node\xspace}
\newcommand{\aps}{nodes\xspace}

\newcommand{\Bss}{Base\-sta\-tions\xspace}

\newcommand{\Aps}{Nodes\xspace}
\begin{document}
\maketitle

\begin{abstract}
In modern wireless networks deployments,  each serving \ap needs to keep its \ac{NCL} constantly  up-to-date to keep track of network changes.
The time needed by each serving \ap to update its \ac{NCL} is an important parameter of the network's reliability and performance. An adequate estimate of such parameter enables a significant improvement of self-configuration functionalities.


  This paper focuses on the update time of \acp{NCL} when an approach of crowdsourced user reports is adopted.
   In this setting, each user periodically reports to the serving \ap information about the set of \aps sensed by the user itself.
  We show that, by mapping the local topological structure of the  network onto states of increasing knowledge, a crisp mathematical  framework can be obtained, which allows in turn for the use of a  variety of user mobility models.
Further, using a simplified mobility model we  show how to obtain useful upper bounds on the expected time for a serving \ap  to gain full knowledge of its local neighbourhood.
\end{abstract}

\section{Introduction}
\label{sec:introduction}
\ac{NCLD} is a core process of modern wireless networks, especially when deployed in an unplanned and decentralised manner like WiFi hotspots and LTE femtocells~\cite{xiao2016dynamic}. In this scenarios each node needs to independently construct the NCL.
Further, appropriate knowledge of network topology, \ie the neighbourhood structure of each node in the network, allows the design of more efficient routing and interference-avoidance algorithms and improved allocation of limited network resources.
In a number of common situations, relying on explicit communication or on a central controller may be impractical or even impossible, for instance, when neighbouring devices belong to a different operator.
Local knowledge of network topology is enough to produce distributed algorithms for channel allocation in WiFi networks, code selection in small cell networks, distributed graph colouring and routing, and also for problems of joint power and channel allocation optimisation (see Section~\ref{sec:related-work}).

Although related to location discovery, the topic of this manuscript is the discovery of existing neighbours without targeting their actual geographical position.
We focus on the process of \ac{NCLD} via
\emph{crowdsourcing}, meaning that the task of detecting and
reporting the existence of conflicting neighbours is delegated to
users. In this framework, each user periodically reports to the serving \ap information about the set of neighbouring \aps observed, see \eg Figure~\ref{fig:antenne}. Exploiting \ac{UE} measurements is appealing because such technique is easy to implement and virtually cost-free. Nevertheless, the information received from \ac{UE} measurements is disregarded by the serving \ap in most implementations~\cite[\textsection 7.4.1]{homenode2011}.

\noindent Keeping the \ac{NCL} updated is fundamental for a number of reasons:
\begin{itemize}
\item Neighbouring cells can be added, removed, or temporarily offline.
\item The handover to a new cell might be problematic whenever it is not contained in the \ac{NCL} of the serving cell.
\item Some cells should not be added to the NCL list because they might reflect spurious measurements, yielding non-reliable handovers.
\item Neighbours with same PCI should be handled with specific solutions.
\end{itemize}
With these reasons in mind, this manuscript studies the time $T$ necessary to achieve confident knowledge of the \ac{NCL} through \ac{UE} measurements.
Estimating $T$ enables optimal tuning of neighbour cell list management schemes.
$T$ is  also important for the design and deployment of decentralised optimisation schemes. A key example is \emph{self-organisation}, a problem where the network nodes need to optimise their configuration without a central controller, \ie relying on local information only.
There exist many fast and efficient decentralised algorithms to self-organise a WiFi/femtocell network;
these algorithms are generally fed with a \ac{NCL}, which needs to be constantly kept up-to-date.
At implementation level, this means that each node needs to periodically estimate with a sufficient level of confidence which nodes of the network are potential conflicting neighbours.
The majority of  decentralised schemes require that each serving \ap needs local knowledge of the local neighbourhood~\cite{kuhn2006complexity,Johansson1999,luby1988removing,szegedy1993locality}, and any attempt to relax this hypothesis comes at the expense of performance, as shown in Section~\ref{sec:related-work}.

In the literature it is usually assumed that neighbourhood
information may be instantaneously acquired~\cite{olofsson1996concept,magnusson1997dynamic,nguyen2010efficient,amirijoo2008neighbor,atawia2012ranked,kim2010self,aziz2010autonomous}, \ie the time $T$ is considered negligible. In fact this assumption
may often not be valid, either because it is necessary to listen to
the channel long enough to get a high-confidence estimation or
because  hidden nodes/second-hop \ac{NCL} need to be
known as well, and thus it is necessary to use communication with users
or other nodes to obtain such information.
When the time needed by each serving \ap to update its \ac{NCL} is larger  than the time to execute the optimisation algorithm, a decentralised approach might not be the best solution.


In a framework where the \ac{NCL} is built via crowdsourced user
reports, our main goal is to rigorously characterise $T$ and study its properties and bounds.
This is a problem that, to the best of our knowledge, has not yet been addressed in the literature.
%
Our main contributions are the following:
\begin{inparaenum}[(i)]
  \item  the problem of user-reports-based \ac{NCLD} is
    stated for the first time  through a crisp mathematical
    formulation;
    \item we introduce a simple mobility model that
      is useful for gaining insight into those situations where
      crowdsourcing via user reports is likely to yield the greatest
      benefit to a decentralised approach;
  \item   we show that this model can  provide an upper bound on the time to topology discovery, thus it  can be used as a design tool (see  Section~\ref{sec:gener-mobil-model}).
\end{inparaenum}
\begin{figure}
  \centering
  \includegraphics[width=0.75\columnwidth]{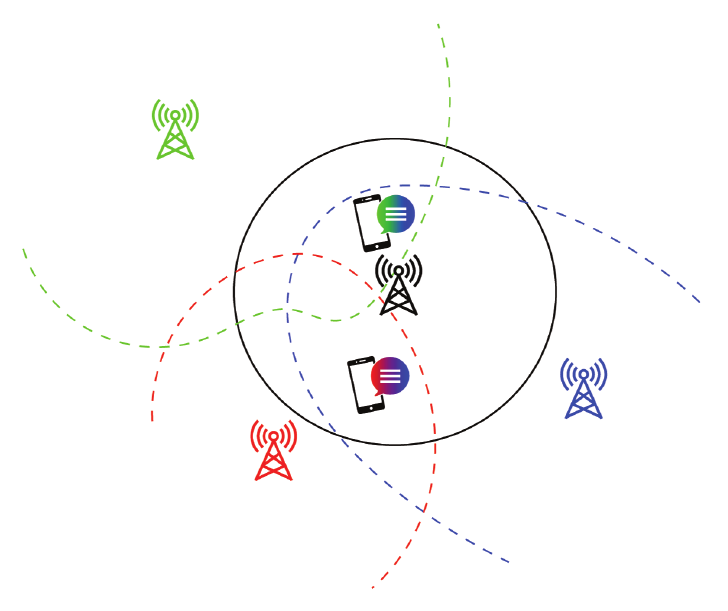}
  \caption{Example of \ac{NCLD}. Each user report to the
    serving \ap which neighbouring \aps they observe.}
  \label{fig:antenne}
\end{figure}

The rest of the paper is organised as follows: in Section~\ref{sec:related-work} we present the related work, then, in Section~\ref{sec:guidelines} we show some practical use cases where our approach can be applied.
In Section~\ref{sec:aclltd-model} we provide a mathematical model for the discovery process and in Section~\ref{sec:problem-statement} we define the problem in the context of such model and give some useful bounds.
We present simulation results to validate the model and show how it can be used as a network design tool in Section~\ref{sec:simulations}. Finally, in  Section~\ref{sec:conclusions} we draw the conclusions.

 \section{Related Work}
 \label{sec:related-work}
In the field of decentralised algorithm design, it has been shown that local knowledge of network topology is enough to produce a distributed algorithm for resource allocation; such local knowledge also allows the minimisation of scrambling-code collision and confusion in small cell networks, see~\cite{checcoself} and references therein.
This knowledge is sufficient, and in a certain sense necessary to build efficient algorithms, as the attempts to relax the hypothesis that each serving node needs knowledge of the local neighbourhood will result in an extreme loss of performances~\cite{duffy2013decentralized, checco2013learning}, that can be prevented only in specific scenarios where the interference model can be described with a simple graph~\cite{checco2017fast}. 

Perhaps the main motivation of this manuscript is the work of~\cite{watanabe2011dynamic}, where the authors propose a neighbour cell list management scheme based on the long-term statistics of UE measurements.
A key parameter of this scheme is the \emph{forgetting factor} $\gamma$, that weigh the longitudinal \ac{UE} measurements. Such parameter is clearly not trivial to tune, especially in settings, very common in wireless and cellular networking, where instantaneous cell list acquisition cannot be assumed. 
The optimal tuning $\gamma$ is only possible by studying the statistical properties of the time $T$ necessary to achieve confident knowledge of the \ac{NCL} through  \ac{UE} measurements. 

The user report function is already available in commercial
femtocells~\cite{patent:EP2214434} and small cells networks, and its implementation for
code confusion and interference reduction is recommended in~\cite{checcoself,edwards2008implementation}.
 
Crowdsourcing approaches have been investigated for different applications, \eg for estimating both density and number of attendees of large events~\cite{EPFL-ARTICLE-187617}.
Many works pertain the use of crowdsourcing for \ac{NCL} discovery. In~\cite{soldani2007self}, the use of mobile measurement to update the \ac{NCL} of macrocells after deployment has been studied: since the intra-frequency reporting function, known as \ac*{DSR},  is energetically costly for the mobile device, the use of it is suggested only in critical situations where a problem with the current \ac{NCL} is known. A similar  case, where the \ac{NCL} needs to be updated when a new macrocell is deployed, is studied in~\cite{parodi2007automatic}.
Other ways to dynamically build the \ac{NCL} via crowdsourcing  are presented in~\cite{hasan2010automatic, watanabe2011dynamic}. A similar work applied to WiMax is presented in~\cite{li2007study}, and a closely related approach for the femtocell case is presented in~\cite{becvar2013dynamic}. 

With this work, we address  the problem of estimating the \ac{NCL} construction time, which is necessary to assess whether crowdsourcing is effective in a particular network deployment. However, to the best of our knowledge, this problem has not been addressed in the literature~yet.

\section{Use Cases}
\label{sec:guidelines}
We show in this section some timely use cases where our proposed framework can be applied as a network design tool.

\subsection{3G Network Optimisation}
\label{sec:self-optimising-3g}

In order to provide seamless mobility and a satisfactory service, the optimisation of the handover function is fundamental in modern 3G cellular networks. To achieve that, the construction of a reliable \ac{NCL} is one of the most critical tasks.
While in the past this was achieved by drive and walk testing, the needs to adapt to changes in the network and to reduce the cost require different solutions~\cite{watanabe2011dynamic,kim2010self,soldani2007self}.

The so-called \ac{DSR} is an intra-frequency 3GPP functionality that allows users to report cells not defined in the \ac{NCL}. In this way, whenever a macrocell detects a problem, or when a new cell is deployed~\cite{parodi2007automatic}, such a function can be activated.
The only disadvantage is that such functionality is energetically costly for the mobile device, so its use is recommended for short periods of time and only in critical situations where a problem with the current \ac{NCL} is known. Therefore, an estimation of the optimal time to keep the \ac{DSR} active is
required. 
Our work provides an effective framework to make such estimation possible.

\subsection{Small cells Self-configuration}
\label{sec:femt-self-conf}
An important problem that affects the small cells deployment for residential use is code selection. In 3G, \bss have only few scrambling codes available, making the task of selecting the optimal allocation challenging.
Moreover, communication with a central controller is discouraged, to avoid signalling overhead.  In 4G and 5G, Physical Cell Identity codes and  5G scrambling codes have similar problems.

A fully decentralised algorithm that can converge to the optimal confusion- and collision-free code allocation has been devised in~\cite{checcoself}. However, it relies on the assumption that small cells are able to construct  their \ac{NCL}. Unfortunately, small cells are often not able to detect first- and second-hop neighbours reliably due to hidden-node effects and the absence of an efficient sniffing \ac{CPICH} mechanism.
A technique to construct the \ac{NCL} via crowdsourcing has been proposed by~\cite{becvar2013dynamic}.
However, the implementation of such a technique would first require the evaluation of the time scale of the \ac{NCL} construction and its comparison with the time scale of the convergence of the  code allocation algorithm.

\section{\acl{NCLD} Model}\label{sec:aclltd-model}
Given a set of wireless \aps $\A = \{a_0,\dots,a_N\}$, let $A(a_i)
\subset \mathbb{R}^2$ denote the coverage area of access point
$a_i$. Please note that $A(a_i)$ generally depends on the
transmission power of $a_i$ and on the radio propagation properties of
the medium.  We focus on serving access point $a_0$ and let $\B$ denote
the neighbouring \aps that have non-void intersection with $A(a_0)$,
\ie
  \[\B = \{a_i \in \A, i>0  \colon A(a_0) \cap A(a_i) \neq \emptyset \}.\]
  We will hereafter use the symbol $N$ to denote the cardinality of
  $\B$, \ie $N = |\B|$.

  Let $\field{P}(\B)$ denote the powerset of $\B$.
  A \emph{tessellation} of the area $A(a_0)$ is the collection of tiles $\{A_{\vec{i}}\}_{\vec{i} \in \field{P}(\B)}$ such that
  \begin{equation}
    \label{eq:1}
    A(a_0)=\bigcup_{\vec{i} \in \field{P}(\B)}A_{\vec{i}}\,,
  \end{equation}
  where
  \begin{align}
    \notag  A_{\vec{i}} =&  \bigcap_{j \in \vec{i}} A(a_j) \cap A(a_0)    \setminus\bigcup_{j \not\in \vec{i}} A(a_j), \quad \vec{i}\neq \emptyset, \\
    A_\emptyset =& A(a_0) \setminus \bigcup_{\vec{i} \in \field{P}(\B)\setminus \emptyset} A_{\vec{i}}.\label{eq:tessellation}
  \end{align}
  In what follows each element $A_{\jvec}$ composing the tessellation
  is referred to as a \emph{tile}, and we  will use the vector
  notation $\jvec$ to represent a set of neighbouring \aps.
  Let us consider for example $\vec{i}=\{a_{1},a_{2}\}$; then, the
  tile $A_{\vec{j}}$ is the portion of $A(a_0)$ that
  is covered by $a_{1}$ and $a_{2}$ only, see Figure~\ref{fig:antenne_b}.
\begin{figure}
  \centering
  \includegraphics[width=0.75\columnwidth]{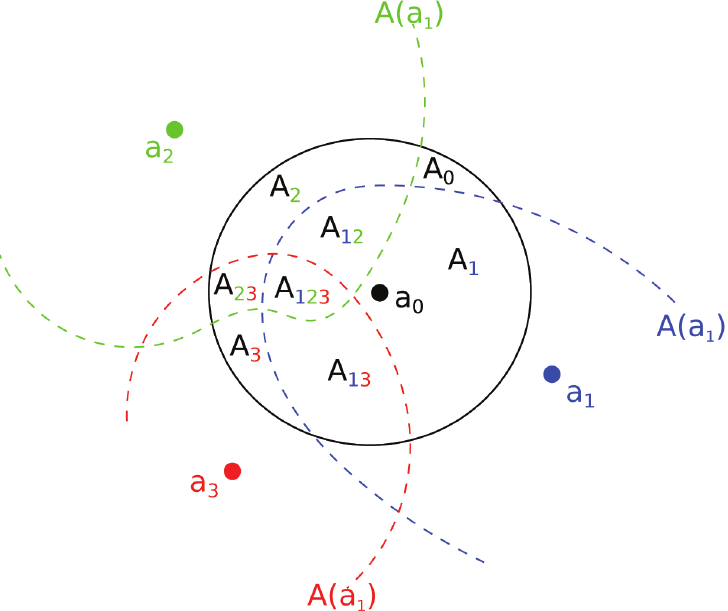}
  \caption{Example of tesselation corresponding to scenario of Figure~\ref{fig:antenne}. The scenario is modelled with a serving \ap, $a_0$, with
    three interfering  neighbours, $a_1,a_2,$ and $a_3$; the coverage
    area of the serving \ap, $A(a_0)$, can be tessellated with the
    sets $A_0, A_1,A_2,A_3,A_{12},A_{13},A_{23},A_{123}$.
  For simplicity of notation, we will write $A_0 \coloneqq A_\emptyset$.}
  \label{fig:antenne_b}
\end{figure}

  \begin{figure}
    \centering
    \includegraphics[width=0.7\columnwidth]{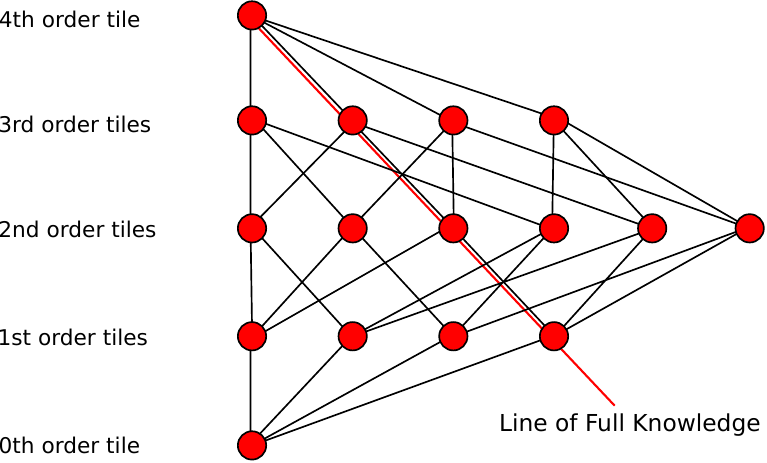}
    \caption{Hypercube representation of the tessellation for $N=4$. There is one zeroth order tile,
      namely $A^C$, four first order tiles, $A_1, A_2, A_3$ and $A_4$, six second order tiles,
      $A_{12}, A_{13}, A_{14}, A_{23}, A_{24}, A_{34}$, four third order tiles, $A_{123}, A_{124}, A_{134}, A_{234}$
      and a fourth order tile, \ie $A_{1234}$.}
    \label{fig:tart1}
  \end{figure}

  Whenever a user is in $A_{\vec{j}}$, it will \emph{report} $\vec{j}$ to  access point   $a_0$. In other words, $a_0$ will be aware of the existence of those
  neighbouring \aps $a_i \in \vec{j}$. The rate  of these
  reports depends on the mobility model assumed (see Section~\ref{sec:problem-statement}).

  To keep the model as conservative as possible, and to
  encompass the frequent case of half-duplex \aps, we assume
  $a_0$ cannot detect the existence of any neighbour even though $a_0$ lies in one of the neighbours coverage area.

  Let $\field{K}_t$ denote the \emph{knowledge set} of access point $a_0$, \ie the set of neighbours that $a_0$ is aware of, at time $t$.
  Given a sequence of reports $\{\vec{j}^1, \dots, \vec{j}^t\}$, we have that
  $\field{K}_t = \bigcup_{i=1}^t \vec{j}^i$. $\field{K}_t$ is a
  sequence of sets that satisfies
  \begin{equation}
    \label{eq:nested}
\field{K}_t = \bigcup_{s=0}^t \field{K}_s;
  \end{equation}
in particular, $\abs{\field{K}_t}$ is non-decreasing in $t$.  Clearly, the knowledge state at time $t$, $\field{K}_t$, take values in $\field{P}(\B)$.
  \begin{definition}[\acl{FK}]
    Given an integer $T$ and a finite sequence of reports
    $\{\vec{j}^1, \dots, \vec{j}^T\}$, the \ap $a_0$ is said to
    have \emph{\acl{FK}} (\acs{FK})
    of its neighbours at time $T$ if
    \[\field{K}_T =\bigcup_{s=1}^T \vec{j}^s = \B.\]
  \end{definition}
  \begin{rmk}
    If $a_0$ has \ac{FK} of its neighbours at time $T$, so it has at
    all times $T+t$ for $t \geq 0$.  In other words, once $a_0$ has reached
    \ac{FK}, it cannot lose it.
  \end{rmk}
  \begin{definition}[First time to \acs{FK}]
    Given a sequence of reports $\{\vec{j}^1, \vec{j}^2, \dots\}$, the
    \emph{first time to \acs{FK}} $\tau$ for the \ap $a_0$ is the first
    time the latter reaches \acs{FK} of its neighbours, \ie
    \begin{equation}
      \label{eq:2}
      \tau \coloneqq \min\{T \geq 0 \text{ such that } \field{K}_T = \B\}\,.
    \end{equation}
  \end{definition}
  \begin{rmk}
    \label{rmk:concerto_no4}
    The characterisation of the first time to \ac{FK} generally depends  on the realisation of a sequence of user reports; this
    means that $\tau$ is a random variable. More precisely,
    by~\eqref{eq:2}, $\tau$ is a \emph{stopping time}, see~\eg
    \cite{levin2009markov}.
  \end{rmk}

  We end the section with a note on the tessellation:
  \begin{rmk}
    \label{rmk:concerto_no5}
    A generic tessellation of $\B$ can be represented as a hypercube
    $H=\{0,1\}^N$ by identifying the vertices of $H$ with the tiles
    $A_\jvec$ that the tessellation is composed of.  The number of
    tiles of a generic tessellation of $\B$ is $2^N$ as well as the
    vertices of a hypercube, represented as vectors of size $N$.  The
    tiles of the tessellation can be mapped onto the vertices of the
    hypercube by identifying the \mbox{$i$-th} component of the
    vertices $x\in H$ with $a_i \in \B$.  In other words,
    \[A_\jvec \leftrightarrow x \quad \Leftrightarrow\quad x_i =
    \mathds{1}_{\{a_i\in\jvec\}} \,,\; i=1,\ldots,N,\] where
    $\mathds{1}$ is the indicator function.  We define the
    \emph{order} of a tile as the number of neighbours a report from
    that tile would give knowledge of; the number of \mbox{$k$-th}
    order tiles is ${N \choose k}$. A report from a \mbox{$k$-th}
    order tile is equivalent to $k$ first-order reports.  In
    particular, \ac{FK} is attained with a report from the
    \mbox{$N$-th} order tile, or at least two reports from two
    distinct \mbox{$(N-1)$-th} tiles, etc.  This property can be
    graphically represented by what we call the \textit{Line of
      \acl{FK}}, see Figure~\ref{fig:tart1}. The line of \ac{FK} is
    clearly not unique\footnote{For example, there are $N$ tiles of
      order $N-1$, but only $2$ are part of a given line of \ac{FK}.};
    the aim of Figure~\ref{fig:tart1} is only to illustrate that a
    sequence of $T$ reports $\{\vec{j}^1, \dots, \vec{j}^T\}$ is a
    path on the hypercube $H$, and that \ac{FK} is attained whenever a
    line of \ac{FK} is reached at a time smaller than $T$.

    Since $\field{K}_t$, the knowledge state at time $t$, takes on
    values in the same set $\field{P}(\B)$, we can also map the
    knowledge states on the hypercube $H$. That is, a sequence
    of reports $\{\vec{j}^1, \vec{j}^2, \dots, \vec{j}^t\}$ is
    equivalent to a single report from tile $\bigcup_{s=1}^t \vec{j}^s =
    \field{K}_t$.
\end{rmk}

We can now define the main problems of this work.
\begin{problem}[Expected first time to \acl{FK}]\label{problem:firsttime}
  Given an access point $a_0$, a set $\B$ of neighbours with given
  position and coverage area, and a sequence of user reports, we want to characterise the expectation
  of the \emph{first time to \ac{FK}}, \ie
  \[
    \mathbb{E}(\tau) = \sum_{t\geq1} t\, \mathbb{P}(\tau=t)\,.
  \]
\end{problem}

Obviously, the way the user(s) moves inside the coverage area $A(a_0)$
heavily affects the difficulty of the problem and its answer.
However, the formulation of Problem~\ref{problem:firsttime} has the
great advantage of decoupling the notion of \ac{FK} from the user
mobility model;
addressing the mean value of the first time to \ac{FK} is also an
enabler to the estimate of the tail of the distribution of $\tau$~--~through Markov's inequality, for example.  Further, from a numerical
point of view, the expected time to \ac{FK} may be achieved via a
Monte Carlo simulation once the set $\B$ and the mobility model in use
are fixed.

  There may exist cases where it is only necessary to characterise the
  first time to attain partial knowledge of the local topology.
  For example, we may be interested in the first moment when the
  neighbouring \aps that have been already discovered, \ie{} the
  elements of the knowledge set $\field{K}_t$, are enough to describe
  a given fraction of the local topology. This idea motivates the following

  \begin{problem}[Expected first time to \knowledge{$\delta$}]
    \label{problem:deltaknowledge}
    Let $\rho$ be a measure over $\mathcal{P}(\B)$ and fixed $\delta
    \in (0,1]$.
    Given an access point $a_0$, a set $\B$ of neighbours with given
    position and coverage area, and a sequence of user reports, we
    want to characterise the expectation of the \emph{first time to
      \knowledge{$\delta$}} $\mathbb{E}(\tau_\delta)$, where
    \[
      \tau_\delta = \min\left\{T\geq0 \text{\emph{ such that }}
        \frac{\sum\limits_{\vec{k} \in \operatorname{\field{P}}(\field{K}_T)}
      \rho\left(A_{\vec{k}}\right)}{
      \sum\limits_{\vec{j} \in \operatorname{\field{P}}(\B)}
      \rho\left(A_{\vec{j}}\right)} \ge \delta \right\}\,.
    \]
  \end{problem}
    When $\delta=1$ and $\rho(A_{\vec{j}}) > 0$ for each $\vec{j} \in
    \operatorname{\field{P}}(\B)$,
    Problem~\ref{problem:deltaknowledge} is equivalent to
    Problem~\ref{problem:firsttime}. Indeed,
    $\nicefrac{\sum\limits_{\vec{k} \in
        \operatorname{\field{P}}(\field{K}_T)}
      \rho\left(A_{\vec{k}}\right)}{
      \sum\limits_{\vec{j} \in \operatorname{\field{P}}(\B)}
      \rho\left(A_{\vec{j}}\right)} \geq 1$ if and only if
    $\field{K}_T\equiv \B$.

    We will hereafter consider the Lebesgue measure
    $\rho(A_{\vec{k}}) = \norm{A_{\vec{k}}}$. This leads to the
    following interpretation: \knowledge{$\delta$} is attained when the
    knowledge set $\field{K}_t$ defines
    for the first time a tessellation that covers
    a fraction of  $A(a_0)$ larger or equal than $\delta$.
    Equivalently, $\tau_\delta$ is the first time when the tiles that
    would give new information\footnote{In the sense that the
      cardinality of the knowledge set $\field{K}_T$ would increase.}
    cover a fraction of $A(a_0)$ that is smaller than $1-\delta$.

  \begin{rmk}
    The concept of \knowledge{$\delta$} is fundamental in the
    simulation phase, when we want to know whether user reports can
    effectively be used to give knowledge of the local
    topology. Indeed, it is likely that the neighbours $a_i$ whose
    coverage area do not overlap with $A(a_0)$ save for a nearly
    negligible portion, will be discovered after a very long time; in
    other words, the leading contribution to $\mathbb{E}(\tau)$ will
    be represented by the mean first visit time of the user(s) to
    $A(a_i)$. Discarding $a_i$ from the picture, the concept of
    \knowledge{$\delta$} let us focus on the quantitative analysis of
     \ac{NCLD}, see Section~\ref{sec:simulations}.
  \end{rmk}

\section{Teleport Mobility}
\label{sec:problem-statement}

  The characterisation of $\tau$, the first time to \ac{FK},
  depends on the assumed user's mobility model:
  it describes how users enter, exit, and move within $A(a_0)$.
The users evolution can then be represented as a pair $U_t=(n_t, X_t)$,
  where $n_t$ is the number of users that lie in $A(a_0)$ at time $t$,
  and $X_t = (x_t^1, x_t^2, \ldots, x_t^{n_t})$ is a vector with
  the position of the $n_t$ users.
  We assume the evolution of $U_t$ to be driven by a discrete-time
  \ac{MC} throughout the paper.

  The realisation of $\{U_t\}_{0\leq t\leq T}$ completely determines
  the sequence of user reports $\{\vec{j}^1, \dots, \vec{j}^T\}$ to
  the access point $a_0$, cf.\ Remark~\ref{rmk:concerto_no4}.
  Since $\field{K}_t$ only depends on $\field{K}_{t-1}$ and $U_t$, then
  the bivariate process $(U_t,\field{K}_t)$ is a \ac{MC}.

    It will prove useful to consider a simplified mobility model in which a
    single user continuously teleport between tiles, without leaving $A(a_0)$\footnote{This model will be extended to many users and to more general models in Section~\ref{sec:gener-mobil-model}}:
\begin{model}{(Teleport Mobility)}
\label{model:rain}
A single user moves within $A(a_0)$ according to a discrete-time
\ac{MC} taking on values in $\mathcal{P}(\B)$. At any time the user cannot abandon the whole region, \ie it is constrained within
$A(a_0)$.  At each step, the user
instantaneously teleports with a probability that is proportional to
the measure of the destination tile\footnote{Note that the actual position within a tile is undefined in this model.}. The destination tile can also be the same tile of previous step, meaning that the user would remain on the same tile during that discrete time step.
Assuming that all tiles are Lebesgue-measurable plane sets,  the transition probabilities are
  \begin{equation}
    \label{eq:poissonrain}
    \Prob{\vec{i} \rightarrow \vec{j}} = \frac{\norm*{A_{\vec{j}}}}{\norm*{A(a_0)}},
  \end{equation}
where $\norm*{\cdot}$ denotes the Lebesgue measure.
\end{model}
\begin{rmk}
  Model~\ref{model:rain} greatly simplifies the
  characterisation of $\tau$, the first time to \ac{FK}.
  Indeed, in this mobility model, $\field{K}_t$ is independent of $U_t$, and
  the sole process $\field{K}_t$ is hence sufficient to
  describe the process of gathering knowledge from the user reports.
  We will hereafter refer to $\field{K}_t$ as the \emph{knowledge chain}.
\end{rmk}

Assuming Model~\ref{model:rain}, we can easily describe the process of
gathering knowledge from user reports as a discrete-time
random walk on the hypercube $H=\{0,1\}^N$ (which we have introduced in
Remark~\ref{rmk:concerto_no5}); having knowledge of
$n$ neighbouring \aps is in fact equivalent to receiving a report from the
$n$-th order tile that give information about all of them.

Let $P(\cdot,\cdot)$ be the transition kernel of the knowledge chain.
If $\vec{k} \not\subseteq \vec{l}$, then~\eqref{eq:nested} guarantees that $P(\vec{k},\vec{l})=0$ because such transition would mean a loss of knowledge.
Conversely, when $\vec{k} \subseteq \vec{l}$, a transition from
$\vec{k}$ to $\vec{l}$ happens if the user moves to a tile that
contains the missing information $(\vec{l}\setminus \vec{k})$ and
does not add more information than that.
Therefore,
\begin{equation}
  \label{eq:P}
 P(\vec{k},\vec{l}) =
 \begin{cases}
    \sum\limits_{\vec{m} \in \field{P}(\vec{k})}  \frac{\norm{A_{\{ \vec{m}
      \cup (\vec{l} \setminus \vec{k}) \}}}}{\norm*{A(a_0)}}    & \text{if
    $\vec{k} \subseteq \vec{l} $},\\
  \qquad 0     & \text{otherwise}.
 \end{cases}
\end{equation}

The following result holds:
\begin{lemma}
  \label{thm:tng}
  The matrix $P$ is upper triangular.
\end{lemma}
  \begin{IEEEproof}
    Let us consider the following partial ordering relation among the states:
    \[\vec{k}\preceq\vec{l} \quad\Leftrightarrow\quad
    \vec{k}\subseteq\vec{l}\,.\]
    By~\eqref{eq:P}, $P(\vec{k},\vec{l})\ne0$ only if
    $\vec{k}\preceq\vec{l}$. Therefore, any mapping
    \[\mathcal{P}(\B)\ni\vec{l} \quad\leftrightarrow\quad
    l\in\{1,2,\ldots,2^N\}\]
    such that \[\vec{k}\preceq\vec{l} \quad \Leftrightarrow \quad
    k\leq l\]
    will put the matrix $P$ into an upper triangular form. In
    particular, we can order the states by increasing cardinality and in
    lexicographic order\footnote{For $N$ neighbouring \aps, \ie with
      $2^N$ different tiles, this would mean the sequence $\{1\},\{2\},\dots,\{N\},\{1,2\},\dots,\{\mbox{N-1},N\},\dots,\{1,2,\dots,N\}$.}.
  \end{IEEEproof}

The explicit computation of the whole matrix $P$
using~\eqref{eq:P} is expensive in general since $P$ is a $2^N\times
2^N$ matrix.
However,  as stated above, $P$ is upper triangular. In Section~\ref{sec:eigenvalues} we show that it is possible to
explicitly characterise its spectrum.
For the reader's reference, Table~\ref{tb:example} shows the matrix
$P$ for $N=3$.
\begin{center}
\begin{table*}[htbp]
\newsavebox\mybox
\savebox\mybox{%
$P = \frac{1}{\norm*{A(a_0)}}  \begin{pmatrix}\norm*{A_{0}} &  \norm*{A_1} & \norm*{A_2} &  \norm*{A_3} &  \norm*{A_{12}}&  \norm*{A_{13}} & \norm*{A_{23}} &  \norm*{A_{123}} \\ 0 & \norm*{A_0} + \norm*{A_1} & 0 &  0 &  \norm*{A_2} + \norm*{A_{12}} &\norm*{A_3} + \norm*{A_{13}}  &  0 & \norm*{A_{23}}+ \norm*{A_{123}}   \\ 0 & 0 & \norm*{A_{0}} + \norm*{A_2} & 0 &  \norm*{A_1} + \norm*{A_{12}}&  0 &  \norm*{A_{3}} + \norm*{A_{23}} & \norm*{A_{13}}+ \norm*{A_{123}} \\ 0 & 0 & 0 &  \norm*{A_{0}} + \norm*{A_{3}} & 0  &  \norm*{A_{13}} +\norm*{A_{1}} &  \norm*{A_{2}} + \norm*{A_{23}}&\norm*{A_{12}}+\norm*{A_{123}} \\ 0 & 0 & 0 &  0 & \norm*{A_{0}} + \norm*{A_{1}} + \norm*{A_{2}} +\norm*{A_{12}}  &  0 &  0 & \norm*{A_{3}} + \norm*{A_{13}} + \norm*{A_{23}}+ \norm*{A_{123}} \\ 0 & 0 & 0 &  0 & 0  &  \norm*{A_{0}} + \norm*{A_{1}} + \norm*{A_{3}} + \norm*{A_{13}} &  0 & \norm*{A_{2}} + \norm*{A_{12}} + \norm*{A_{23}} +\norm*{A_{123}} \\0 & 0 & 0 &  0 & 0  & 0 &   \norm*{A_{0}} + \norm*{A_{2}} + \norm*{A_{3}} + \norm*{A_{23}} & \norm*{A_{1}} + \norm*{A_{12}} +\norm*{A_{13}} + \norm*{A_{123}} \\ 0 & 0 &  0 & 0 & 0 & 0  &  0 & 1\end{pmatrix}.$%
}
\newlength\TableWidth
  \settowidth\TableWidth{\usebox\mybox}
  \ifnum\TableWidth>\linewidth
    \setlength\TableWidth{\linewidth}
  \fi
  \resizebox{\TableWidth}{!}{\usebox\mybox}
\caption{Example of transition matrix $P$ for $N=3$.}
\label{tb:example}
\end{table*}
\end{center}

\subsection{Expected Time to \acl{FK}}
\label{sec:expect-time-full}

Let $\vec{k}^\ast=\{1,2,\ldots,N\}$ be the state of \ac{FK}.  By
formula~\eqref{eq:P}, $P(\vec{k}^\ast,\vec{k}^\ast)=1$. This means
that the chain has an absorbing state, and  the hitting time  of
this state is just $\tau$, the first time to \ac{FK}.  Hence, we can
compute the expected time to \ac{FK} simply by
\begin{equation}
  \label{eq:method1}
  \Exp{\tau} = (I-Q)^{-1}\mathbf{1},
\end{equation}
where $Q$ is obtained from $P$ by removing the row and the column
relative to state $\vec{k}^\ast$ and
$\mathbf{1}$ is the column vector of ones~\cite{Gehring1976}.
In a similar way, it is possible to compute the other moments of $\tau$.

Even if $I-Q$ is upper triangular and can be block decomposed, the
computation of its inverse may not be affordable
when the cardinality of $\B$ grows.
In Section~\ref{sec:conv-with-high} we will bound the probability of
the event $\{\tau > t\}$.

\subsection{Expected Time to \knowledge{$\delta$}}
\label{sec:expected-time-delta}
Regarding Problem~\ref{problem:deltaknowledge}, we
can easily modify matrix $P$ to obtain
the expected time to \knowledge{$\delta$}.
Every state $\vec{k} \in \operatorname{\field{P}}(\B)$ such that
 \[ \sum\limits_{\vec{l} \in \operatorname{\field{P}}(\vec{k})}
 \frac{\norm{A_{\vec{l}}}}{\norm*{A(a_0)}} \ge \delta \] can be aggregated in the
 absorbing state, summing the corresponding column of $P$ in the last
 column, and then eliminating the column and row corresponding to
 state $\vec{k}$.
In this way it is possible to compute $\Exp{\tau_\delta}$
using~\eqref{eq:method1}.

\subsection{Eigenvalues}
\label{sec:eigenvalues}
The following result fully characterises the spectrum of the matrix
$P$:
\begin{theorem}
  For $\vec{k} \in \mathcal{P}(\B)$, the eingenvalues of $P$ have the
  form
  \[
  \savebox\mybox{%
    $\lambda_{\vec{k}}
    =\frac{1}{\norm*{A(a_0)}}\bigg(\norm*{A_0}+
    \sum\limits_{\substack{\vec{l}\subset\vec{k}\\\abs*{\vec{l}}=1}}\norm*{A_{\vec{l}}}
    +
    \dots+\sum\limits_{\substack{\vec{l}\subset\vec{k}\\\abs*{\vec{l}}=m}}\norm*{A_{\vec{l}}}
    + \dots
    +\sum\limits_{\substack{\vec{l}\subset\vec{k}\\\abs*{\vec{l}}=\abs*{\vec{k}}}}
    \norm*{A_{\vec{l}}}\bigg).$%
  }
  \settowidth\TableWidth{\usebox\mybox}
  \ifnum\TableWidth>\linewidth
  \setlength\TableWidth{\linewidth}
  \fi
  \resizebox{\TableWidth}{!}{\usebox\mybox}
  \]
  \end{theorem}
\begin{IEEEproof}
  The matrix $P$ being upper triangular by Lemma~\ref{thm:tng}, the entries
  $P(\vec{k},\vec{k})$ are the eigenvalues of the matrix.
  Let us then imagine to have the knowledge chain in state $\vec{k}$.
  The only way for the chain to undergo a self-transition
  ($\vec{k}\to\vec{k}$) is that the user reports any combination of
  neighbouring \aps that have already been discovered. In other
  words, the knowledge chain undergoes a self-transition if and only
  if the user reports an element of $\mathcal{P}(\vec{k})$.
  Therefore,
  \[\lambda_{\vec{k}} = \frac{1}{\norm*{A(a_0)}}\sum_{\vec{l} \in
    \mathcal{P}(\vec{k})}|A_{\vec{l}}|\,.\]
  Last formula is equivalent to the thesis.
\end{IEEEproof}

Since each eigenvalue is a sum of positive elements,
the second-largest eigenvalue $\tilde{\lambda}$ can be obtained
by maximising over the tiles of order $N-1$:
\begin{equation}
  \label{eq:lambda}
\tilde{\lambda} = \max_{\vec{k}\colon \abs*{\vec{k}}=N-1 } \lambda_{\vec{k}}.
\end{equation}

\subsection{Convergence Properties, Bounds}
\label{sec:conv-with-high}

Using~\eqref{eq:lambda}, it is possible to obtain the following result:
\begin{lemma}
\label{cor:bound}
Given $\epsilon>0$, let
\begin{equation}
  \label{eq:bound}
  S(1-\epsilon) = \frac{\log{\epsilon}}{\log{\tilde{\lambda}}}.
\end{equation}
Then, $S(1-\epsilon)$ reports are sufficient to achieve \ac{FK}
with probability greater or equal than $(1-\epsilon)$.
\end{lemma}
\begin{proof}
  Using Lemma~\ref{thm:tng} and Equation~\eqref{eq:lambda} on $P$,
  \begin{equation}
    \label{eq:6}
    \Prob{\tau > t} \le \tilde{\lambda}^t.
  \end{equation}
For a small target tolerance $\epsilon$ of not achieving \ac{FK},
\begin{equation}
  \label{eq:method4}
  \text{if } t\ge \frac{\log{\epsilon}}{\log{\tilde{\lambda}}}
  \quad \Rightarrow \quad \Prob{\tau > t}\le \epsilon.
\end{equation}
\end{proof}

\subsubsection*{\knowledge{$\delta$} Convergence Bounds}
\label{sec:delta-knowl-conv}
Using the same manipulation of the matrix $P$ described in
Section~\ref{sec:expected-time-delta}, Corollary~\ref{cor:bound} can
be applied to the modified matrix to obtain a bound for the
number of steps to have \knowledge{$\delta$} with high probability.

\subsection{\knowledge{$\delta$} and other-than-teleport mobility}
\label{sec:gener-mobil-model}

  Model~\ref{model:rain}  is equivalent to a single user teleporting
  instantaneously to a random point within the coverage area of the
  \ap; time is discrete.
  Thus, at each time the \ap $a_0$ receives a user report from
  a point drawn according to the uniform probability distribution over
  the coverage area $A(a_0)$.
   Having in mind the numerical characterisation of the first
    time to  \knowledge{$\delta$}, the teleport model is
    particularly convenient. This task
    could be in fact carried out within the Monte Carlo paradigm by
    simply throwing sufficiently many
    points at random inside the coverage area $A(a_0)$. In other
    words, it is possible to numerically study the process through
    which  \knowledge{$\delta$} is achieved by sampling
    sufficiently many times a probability density function that is
    uniform over the coverage area $A(a_0)$.

  Model~\ref{model:rain} may prove itself unsatisfactory in a real
  life scenario. The main problem is that
  if we generate a sequence of user reports according to it, any two
  elements of the sequence are independent, whereas in general they
  are not. In each mobility model where the trajectory taken by the
  user is physically feasible, the user positions communicated by two
  successive reports are in fact correlated due to the motion
  constraints.

  Let us imagine that a single user travels inside the coverage area
  $A(a_0)$ according to an unknown mobility model, and let $T(t)$ be the
  trajectory taken by the user. Sampling the
  trajectory at equally-spaced discrete times, we obtain an
  embedded sequence of user locations, which correspond to an embedded
  sequence of user reports. Next, we can analyse the sequence and
  understand after how many steps  \knowledge{$\delta$} has
  been reached. By multiplying this number of steps by the time lapse between
  two consecutive reports (inter-report time), the time to
  \knowledge{$\delta$} can be obtained for that particular realisation
  of the user-reports sequence. Finally, the procedure above can be repeated
  sufficiently many times to estimate with a Monte Carlo method the
  expected time necessary for $a_0$ to reach
  \knowledge{$\delta$}.

  As mentioned above, in a general mobility model it is likely that
  two successive user reports are correlated.
  These correlations may decay as the inter-report time grows larger
  and larger. As an example, let us
  imagine that a single user travels inside the coverage area $A(a_0)$
  according to a \ac{MC}. Let $\pi$ be the equilibrium probability
  measure of the chain and let $t_{\text{mix}}$ be the mixing time of the
  chain, \ie the time needed for the chain to reach equilibrium.
  If the inter-report time is chosen comparable to $t_{\text{mix}}$
  then the time lapse between two successive reports will be
  sufficient for the \ac{MC} to forget the past trajectory; in other
  words, the correlations between consecutive reports will be
  negligible. As a
  consequence, the user locations will be independently drawn from the
  probability measure $\pi$, and the matrix $P$ describing the
  knowledge evolution will become
  \begin{equation}
    \label{eq:carcarlo}
    \tag{\ref*{eq:P}'}
    P(\vec{k},\vec{l}) =
    \begin{cases}
      \sum_{\vec{m}\in\mathcal{P}(\vec{k})}
      \pi\left(A_{\{\vec{m}\cup(\vec{l}\setminus\vec{k})\}}\right)
      \quad& \text{if } \vec{k}\subseteq\vec{l},\\
      0&\text{otherwise.}
    \end{cases}
  \end{equation}

  Therefore, the formulation and the results developed in
  Sections~\ref{sec:expect-time-full}--\ref{sec:conv-with-high} are
  still valid if we consider a single-user mobility model based on a
  \ac{MC}, provided that the time lapse between two consecutive
  reports is of the order of the mixing time of the chain.  Under the
  assumption that user reports are sent at a frequency comparable with
  the inverse mixing time of the mobility \ac{MC}, we can compute an
  upper bound on the time to \knowledge{$\delta$}. Any reporting rate
  higher than $\tfrac{1}{t_\text{mix}} $ will in fact still guarantee
  that $a_0$ achieves \knowledge{$\delta$} of its neighbourhood in at
  most $\Exp{\tau_\delta}\cdot t_\text{mix}$ seconds on average\footnote{Recall that $\Exp{\tau_\delta}$ is measured in number of reports.}.

  \subsubsection{Multi-user Scenario}
  \label{sec:multi-user-scenario}
    We end this section by briefly mentioning a straightforward
    application of Model~\ref{model:rain} in a multi-user scenario.
    Let us imagine that $n$ users may enter, move within, and exit
    $A(a_0)$ according to a hidden mobility model.
    We assume that $n$ is a very large number and that
    it is possible to statistically characterise the stationary
    user-density by means of a probability measure $\pi$ over $A(a_0)$.
    At each time every user may independently send a report with a very small
    probability $p$.
    Then, the number of reports received by $a_0$ in a given time interval is
    approximately Poissonian and the time lapse between two successive reports is
    exponential with parameter $\lambda = np$.
    Next, let $m=\Exp{\tau}$ be the expected time to \ac{FK}, expressed in number of
    reports, returned by~\eqref{eq:carcarlo} and~\eqref{eq:method1};
    the expected time to achieve \ac{FK} is the expectation of the
    first time for a Poisson process of parameter $\lambda$ to hit the
    state $m$.
    A practical example for this kind of scenario in presented in Section~\ref{sec:cells-depl-cong}.

\subsection{Examples}
\label{sec:examples}

\subsubsection{Femtocells Deployment for Residential Use}
\label{sec:femt-depl-resid}

Regarding the use case of femtocell self-organisation presented in Section~\ref{sec:femt-self-conf}, each \bs
 serves a very small number of devices.
 Using data of typical residential densities and coverage areas, a statistic of the tessellation can be devised. If it is possible to establish a time $\tilde{T}$ after which the user position can be considered as drawn
from a uniform distribution, then $S(\delta)$ is an upper bound of
the time to \knowledge{$\delta$} for all the inter-report times
smaller than or equal to $\tilde{T}$.

\subsubsection{Cells Deployed in Congested Areas}
\label{sec:cells-depl-cong}
Opposite to the previous example, cells deployed in congested
places like a mall have an extremely large basin of potential users.
However, in situations where users main interest is other than
connecting to the internet, it is reasonable to expect the single-user
reporting-activity to be rather sporadic.
Therefore, the Poissonian approximation that we have mentioned at the
end of Section~\ref{sec:gener-mobil-model} may be applicable.
In this case, characterising the time to achieve \knowledge{$\delta$}
is possible through a statistic of the typical (or worst case) tessellations.

\section{Simulations}
\label{sec:simulations}
\subsection{Teleport Model on Random Positioned \Aps}
\label{sec:simple}
In this section we offer a preliminary assessment of the possibility
to use the machinery developed so far in real applications.  To this
purpose, we developed a simulation framework in MATLAB and studied a
scenario where \num{8} \aps are positioned on a plane at random
according to a uniform (bivariate) probability distribution,
i.e. \emph{uniformly at random}.
Each \ap has a circular coverage area of the same size.  We considered
\num{350} different configurations, with the constraint that the
coverage area of $a_0$ has non-void intersection with the coverage
area of the remaining \aps, meaning that \ac{FK} is achieved as soon
as all \num{7} neighbours are reported to $a_0$.  We compute the
tessellation of each configuration using a classical Monte Carlo
sampler.  For each of these \num{350} configurations, we computed the
expected time to \knowledge{\num{0.9}} $\Exp{\tau}$ together with the
number of steps sufficient to guarantee \knowledge{\num{0.9}} with
\SI{90}{\percent} confidence, \ie $S(0.9)$.  The inter-report time
being fixed during this first experiment, the amount of time in
seconds to achieve \knowledge{\num{0.9}} is directly proportional to
the number of steps just evaluated.

Figure~\ref{fig:7moving} displays the empirical probability
mass function of these two quantities. $\Exp{\tau}$ is centered
around \num{10} steps, while $S(0.9)$ is shifted on higher values, as
expected being an upper bound.
\begin{figure}[htbp]
  \centering
  \includegraphics[width=0.6\columnwidth]{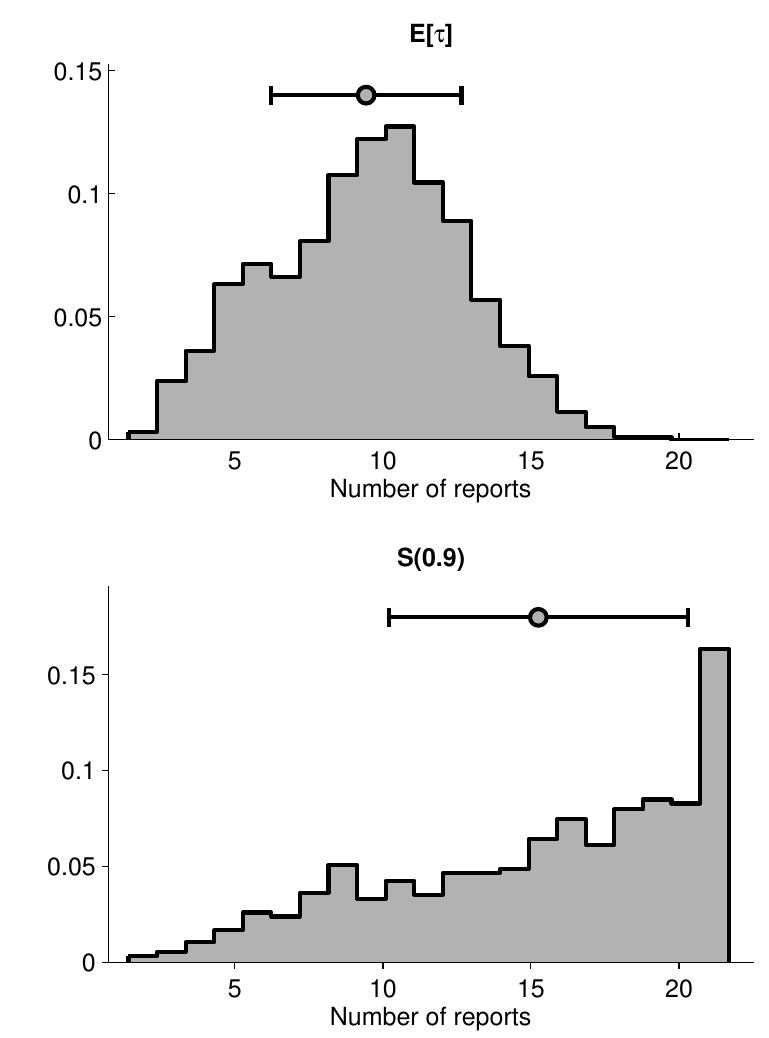}
  \caption[\knowledge{\num{0.9}} distribution]{Empirical probability
    mass function of the expected time to \knowledge{\num{0.9}}, and
    the number of steps to have \knowledge{\num{0.9}} with
    \SI{90}{\percent} confidence, for the teleport model on random
    positioned \aps. Since the inter-report time is fixed, the
    simulation time is directly proportional to the number of reports.}
  \label{fig:7moving}
\end{figure}
\begin{figure}[htbp]
  \centering
  \includegraphics[width=0.7\columnwidth]{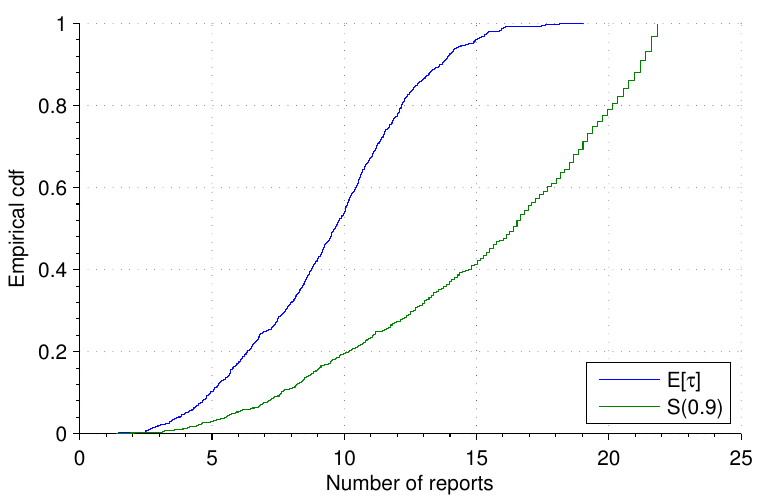}
  \caption[\knowledge{\num{0.9}}, empirical cumulative distribution
  function]{Empirical cumulative distribution function of the expected
    time to \knowledge{\num{0.9}}, and the number of steps to have
    \knowledge{\num{0.9}} with \SI{90}{\percent} confidence, for the
    teleport model on random positioned \aps. Since the inter-report time is fixed, the
    simulation time is directly proportional to the number of reports.}
  \label{fig:7movingCDF}
\end{figure}
Figure~\ref{fig:7movingCDF} shows the empirical cumulative distribution
function of $\Exp{\tau}$ and $S(0.9)$.
We see that \num{16} steps are sufficient to achieve
\knowledge{\num{0.9}} for nearly all scenarios (\SI{95}{\percent}), while we need  \num{22} steps using $S(0.9)$.
We also notice that the bound obtained from~\eqref{eq:bound} is a
conservative estimation, because it uses only the second-largest
eigenvalue $\tilde{\lambda}$. Indeed, it takes into account  only the
slowest way to reach the desired knowledge, while the problem has a
rich combinatorial structure that cannot be completely captured
by~\eqref{eq:6}.

Roughly speaking, a user moving at \SI{0.5}{m/s} according to a random walk model, and providing \emph{at least} one report every hour, can guarantee the \ap will have \knowledge{\num{0.9}} with high probability in less than \SI{5}{h}, and in less than two hours if reports are sent \emph{at least} every 15 minutes (see next section for a more detailed analysis on the interaction between report frequency and our bound). If the local topology is not typically expected to change often, these are acceptable times.

To summarise, simulation on random scenarios show that our proposed bound can be used to estimate the time to \knowledge{$\delta$}. Using realistic values, the expected time to \knowledge{$\delta$} is reasonably small.

\subsection{Random Walk on a Grid}
\label{sec:random-walk-grid}
In order to investigate and confirm the ideas of
Section~\ref{sec:gener-mobil-model}, we simulated the reports sent
with different inter-report times by a random walker that moves within
$A(a_0)$ under the condition of reflective boundary, and compared this
 mobility model with Model~\ref{model:rain} (see
Section~\ref{sec:problem-statement}) for a set of 8 \aps
positioned as described at the beginning of this section.
\begin{figure}[htbp]
  \centering
  \includegraphics[width=0.7\columnwidth]{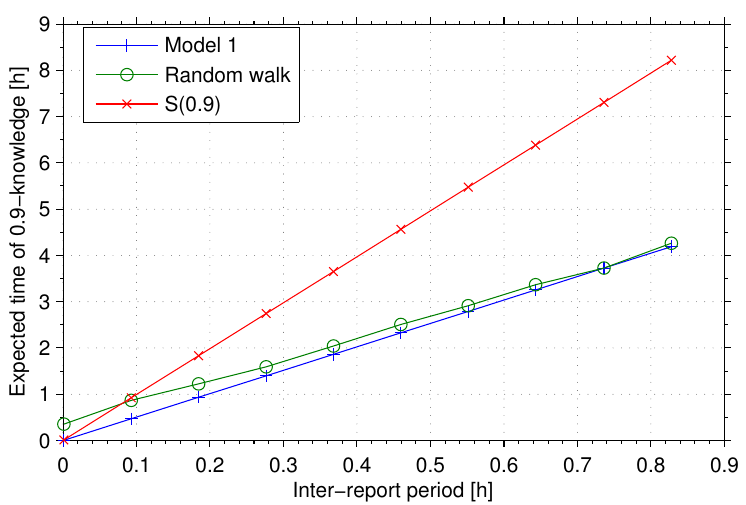}
  \caption[Random Walk]{Empirical mean \knowledge{\num{0.9}} time (in
    hours) of a random walk vs. inter-report period, compared with
    Model~\ref{model:rain} and its bound $S(0.9)$ in a femtocell grid
    of 8 \aps.}
  \label{fig:mobility}
\end{figure}


 In Figure~\ref{fig:mobility} we let the inter-report time  increase and compare the average time to achieve  \knowledge{\num{0.9}} according to both the random walk (green line) and the  teleport model (blue).
  We see that that, if the  inter-report time is sufficiently large, the empirical mean time to achieve  \knowledge{\num{0.9}} for the random walk model is well    approximated by that of Model~\ref{model:rain}.

We assume typical femtocell parameters, \ie that coverage
radius is \SI{50}{\metre} and that the user do a step in a grid of
\SI{2.5}{\metre} every \SI{5}{\second}.
Figure~\ref{fig:mobility} also shows that when reports are sent each
\SI{6}{min} or less, the time to \knowledge{0.9} is smaller than
\SI{1}{h}, but at such high frequency the bound $S(0.9)$ (red line) is not valid anymore.
The reason why more reports than Model~\ref{model:rain} are needed in
the case of high-frequency reports is the following: since the
inter-report time is short, it is likely that many reports will be
sent from the same tile, \ie the knowledge chain will undergo many
self-transitions.

It is important to notice that the inter-report time used in
Figure~\ref{fig:mobility} are far from the theoretical order of magnitude of the
random walk mixing time. Yet Figure~\ref{fig:mobility} suggests that,
for a family of scenarios, it should be possible to determine the value of the inter-report time such that the average time to achieve
\knowledge{\num{0.9}} may be well predicted by Model~\ref{model:rain}.
Once that value of the inter-report time is found, the value of $\mathbb{E}[\tau_{0.9}]$ returned by Model~\ref{model:rain} may serve as an
upper bound to the actual time to achieve \knowledge{\num{0.9}} when
smaller inter-report times are implemented.

To summarise, simulation on random walks corroborate the analysis of Section~\ref{sec:gener-mobil-model}

\subsection{A Realistic Scenario}
\label{sec:realistic-scenario}
A received power map for 4 \bss in the Hynes convention centre
have been generated using the Wireless System Engineering
(WiSE)~\cite{fortune1995wise} software, a comprehensive 3D ray tracing
based simulation package developed by Bell Laboratories.
\begin{figure}[htbp]
  \centering
  \includegraphics[width=0.7\columnwidth]{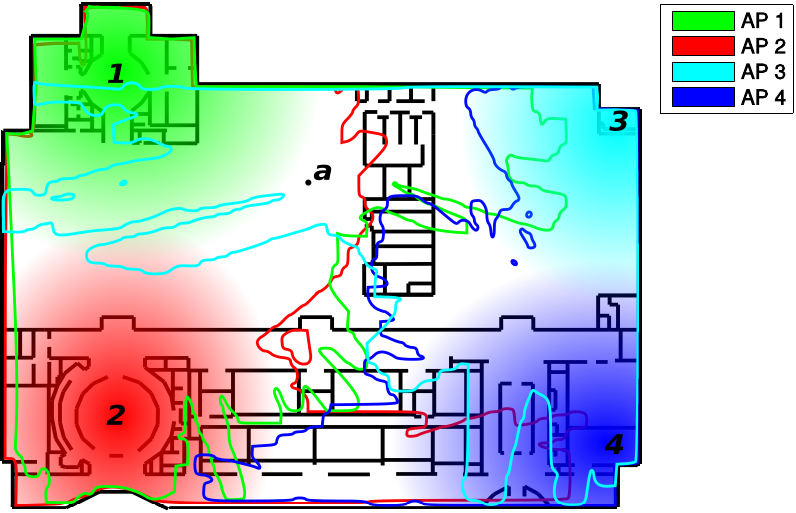}
  \caption[Apartment Example]{Coverage areas at Hynes convention
    centre. The coloured lines delimit the extension of the coverage
    areas.
    Basestations are transmitting at \SI{2.1}{\giga\hertz} with a
    power of \SI{34}{\milli\watt}. Point $a$ lies in the tile $A_{123}$.
  }
  \label{fig:wise}
\end{figure}
\Bss are assumed transmitting at a frequency of \SI{2.1}{\giga\hertz}
with a power of \SI{34}{\milli\watt}. We assume there is a macrocell
that covers the whole building, and we estimate its time to full
knowledge. As before, a Monte Carlo simulation has been made to
estimate the tessellation, and then the expected time to
\knowledge{$\delta$} has been computed using a teleport mobility model
(Model~\ref{model:rain}),
as explained in Section~\ref{sec:expected-time-delta}.

Figure~\ref{fig:wise} shows the corresponding coverage areas
when the power detection threshold is \SI{-70}{\dBm}. Although the shape
of the coverage areas and their intersection is much
more complex than the simple scenario depicted in
Section~\ref{sec:simple}, it is still possible to construct the
tessellation by considering which coverage areas each spatial point
lies in. For example, point $a$ lies in the coverage area of \aps $1$,
$2$, and $3$, so it belongs to the tile $A_{123}$.

Figure~\ref{fig:wise_delta} displays the expected time to
\knowledge{$\delta$}, $\Exp{\tau_\delta}$ when $\delta$ is varied.
We notice a step-function-like behaviour, with a new step that is
added every time a new
state become absorbing, as explained in
Section~\ref{sec:expected-time-delta}.
\begin{figure}[htbp]
  \centering
  \includegraphics[width=0.7\columnwidth]{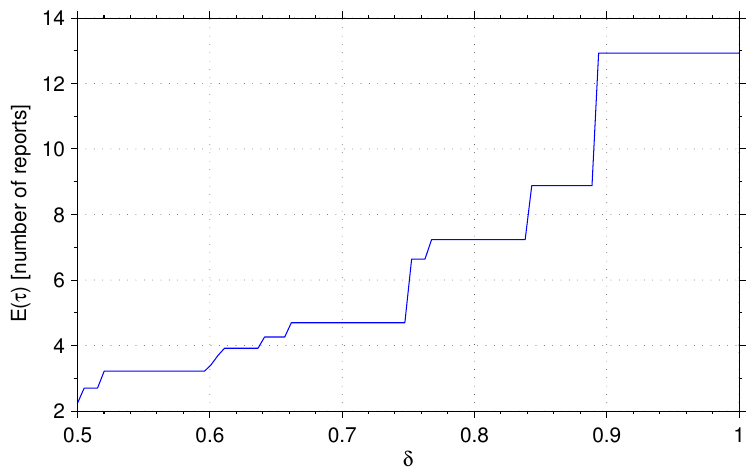}
  \caption[\knowledge{$\delta$} varying $\delta$]{Expected time to
    \knowledge{$\delta$}, $\Exp{\tau_\delta}$ using
    Model~\ref{model:rain} in  Hynes convention centre for different
    values of    the parameter $\delta$. The inter-report time is set
    to the very same value used for Figures~\ref{fig:7moving} and~\ref{fig:7movingCDF}.}
  \label{fig:wise_delta}
\end{figure}
\begin{figure}[htbp]
  \centering
  \includegraphics[width=0.7\columnwidth]{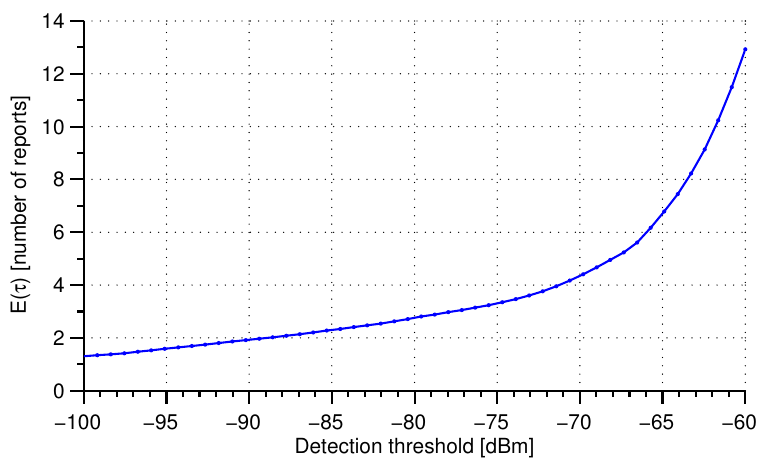}
  \caption[\knowledge{$\delta$} varying detection threshold]{Expected
    time to \ac{FK}, $\Exp{\tau}$, using Model~\ref{model:rain} in
    Hynes convention centre,  for different values of the
    user-detection threshold. The inter-report time is set
    to the very same value used for Figures~\ref{fig:7moving}
    and~\ref{fig:7movingCDF}. }
  \label{fig:wise_threshold}
\end{figure}

Figure~\ref{fig:wise_threshold} shows the behaviour of $\Exp{\tau}$,
the expected time to \ac{FK}, when the user detection threshold varies from a
very conservative value of \SI{-60}{\dBm} to a more realistic one of
\SI{-100}{\dBm}. When the users are more sensitive, the coverage
areas, and specifically the higher order tiles, are bigger, leading
to better performance.  In particular, we  see that an average of 14 steps are enough
to achieve \ac{FK}.

To summarise, these results seem to confirm that the values obtained placing random nodes with circular coverage areas in Section~\ref{sec:simple} are compatible with real world scenarios, so the use of statistics obtained from macroscopic parameters as densities of deployment and distribution of coverage radii can be used as a tool to bound the time to \knowledge{$\delta$}.

\section{Conclusions}
\label{sec:conclusions}

In this paper, we have introduced the problem of user-reports-based \acl{NCLD} and provided a
crisp mathematical formulation of it for a simple mobility model. We have also shown that such mobility model can be effectively used as an
upper bound for a wide range of mobility models when the user reports
frequency is lower than the inverse mixing time of the \acl{MC} of the actual mobility model. Additionally, we have provided a useful method to
estimate the time to \knowledge{$\delta$} when the problem is too complex
to be solved exactly.

Simulations on random scenarios with typical small cells parameters show
that the expected number of reports in order to have a high degree of
knowledge of the local topology is very small.  Roughly speaking, a
user moving at \SI{0.5}{m/s} according to a random walk model, and
providing at least one report every hour, can guarantee the serving \ap will have \knowledge{\num{0.9}} with high probability in less than \SI{5}{h}, and in less than two hours if reports are sent at least every 15 minutes. Since we do not expect the network topology to be affected by high network dynamics, these are acceptable times for the problems of interest.
We encourage the adoption of the presented framework to assess the possibility of employing crowdsourced user reports in other self-configuration problems, comparing the time to \knowledge{$\delta$} with the expected time to convergence of a given decentralised algorithm.

Simulations in more realistic scenarios show that the bounds obtained are compatible with the ones obtained from statistics on random scenarios with similar parameters. This seems to confirm that the use of statistics obtained from macroscopic parameters, such as densities of deployment and distribution of coverage radii, can be used as a tool to bound the time to \knowledge{$\delta$}.

In conclusion, we provide a useful tool to estimate the time to \ac{NCL} construction, which is fundamental to assess whether a decentralised algorithm can be employed in a given network scenario.

\section{Conflict of Interest}
\label{sec:coi}
The authors declare that there is no conflict of interest to disclose regarding the publication of this paper.

\section{Acknowledgments}
\label{sec:acknowledgments}
We would like to thank Anna Zakrzewska, from Bell Laboratories, Alcatel Lucent Ireland, for her insights and useful suggestions.

\addcontentsline{toc}{section}{References}
\bibliographystyle{IEEEtran}
\bibliography{CL}

\end{document}